\newcommand{\version}{November 18, 2011}     

\documentclass[12pt]{article}
\usepackage{bbm}
\usepackage{a4,amsthm,amsfonts,latexsym,amssymb}
\usepackage{mdwlist}
\usepackage{graphicx}

\swapnumbers 

\theoremstyle{plain}
\newtheorem{thm}{THEOREM}
\newtheorem{cor}[thm]{COROLLARY}
\newtheorem{lem}[thm]{LEMMA}
\newtheorem{define}[thm]{DEFINITION}
\newtheorem{proposition}[thm]{PROPOSITION}

\newcommand{\beq}{\begin{equation}}
\newcommand{\eeq}{\end{equation}}
\def\beqa{\begin{eqnarray}}
\def\eeqa{\end{eqnarray}}

\newcommand{\non}{\nonumber}

\newcommand{\C}{{\mathbb C}}
\newcommand{\N}{{\mathbb N}}
\newcommand{\R}{{\mathbb R}}
\newcommand{\Z}{{\mathbb Z}}

\newcommand{\one}{{\mathbbm 1}}
\newcommand{\Tr}{{\rm Tr}}

\newcommand{\range}{{\rm range}}

\newcommand{\lspan}{{\rm lin. span}}

\newcommand{\Hh}{{\mathcal H}}
\newcommand{\Rr}{{\mathcal R}}
\newcommand{\A}{{\mathcal A}}

\newcommand{\M}{{\mathcal M}}

\newcommand{\U}{{\mathcal U}}

\newcommand{\D}{{\mathcal D}}
\newcommand{\TT}{{\mathcal T}}
\newcommand{\V}{{\mathcal V}}
\newcommand{\W}{{\mathcal W}}
\newcommand{\X}{{\mathcal X}}
\newcommand{\eps}{\varepsilon}

\renewcommand{\S}{{\mathcal S}}

\newcommand{\tdt}{\!\cdot\!}

\newcommand{\fe}{${\rm f.e. \,}$}

\def\bce{\begin{center}}
\def\ece{\end{center}}
\def\bit{\begin{itemize}}
\def\eit{\end{itemize}}

\date{\small\version}
\begin{document}
\markboth{\scriptsize{State Space Structures and Quantum Dynamical Semigroups \qquad  \version}}
         {\scriptsize{State Space Structures and Quantum Dynamical Semigroups \qquad  \version}}

\title{
\vspace{-80pt}
\bf{The structures of state space concerning Quantum Dynamical Semigroups}}
\author{\vspace{30pt} Bernhard Baumgartner$^1$, Heide Narnhofer$^2$
\\
\vspace{-4pt}
\small{Fakult\"at f\"ur Physik, Universit\"at Wien}\\
\small{Boltzmanngasse 5, A-1090 Vienna, Austria}}

\maketitle


\begin{abstract}
Each semigroup describing time evolution of an open quantum system
on a finite dimensional Hilbert space is related to a special structure of this space.
It is shown how the space can be decomposed into orthogonal subspaces: One part is related to decay,
some subspaces of the other subspace are ranges of the stationary states.
Specialities are highlighted where the complete positivity of evolutions is actually needed for analysis,
mainly for evolution of coherence.
Decompositions are done the same way for discrete as for continuous time evolutions,
but they may show differences:
Only for discrete semigroups there may appear cases of sudden decay and of perpetual oscillation.
Concluding the analysis we identify the relation of the state space structure to the processes
of Decay, Decoherence, Dissipation and Dephasing.
\\[20ex]
Keywords: open system, time evolution, Lindblad generator, semigroup, Schr\"{o}dingers cat
\\[3ex]
PACS numbers: \qquad  03.65.Yz , \quad 05.40.-a , \quad 42.50.Dv ,
\quad 03.65.Fd

\end{abstract}

\footnotetext[1]{\texttt{Bernhard.Baumgartner@univie.ac.at}}
\footnotetext[2]{\texttt{Heide.Narnhofer@univie.ac.at}}


\newpage

\section{Introduction}\label{intro}

Time evolutions of open quantum systems are not appearing as unitary maps,
in general they are irreversible.
The background of their emergence from the fundamental laws of quantum mechanics
is reviewed in Section \ref{setting}.
There is no relation to a generating Hamiltonian acting in the open system alone, so
algebraic aspects have to be studied in a new way.
The Hilbert space of pure state vectors is now insufficient, it is not the stage for time evolution in the Schr\"{o}dinger picture,
when the environment is not included in the description.
One has to consider the set of all states instead, both pure and mixed ones,
since mixing properties of states are in general changing during ``dissipative'' evolution.
This new aspect makes a fundamental difference to the evolution of closed systems.
In the Heisenberg picture such a time evolution does not preserve the complete algebraic structure
of the set of observables. Linear relations are preserved in time, but
commutation relations, functional expressions and product formulas may change.

In spite of loosing the basic properties of unitary time evolution
there remains a coarse grained relation of existing stationary states and of invariant operators
to a decomposition of the Hilbert space into mutually orthogonal subspaces.
For continuous time -- for which
the generators of completely positive norm continuous semigroups,
representing such time evolutions of open systems with a finite dimensional Hilbert space
have been characterized in  \cite{L76, GKS76} --
we have already found such relations:
In \cite{BNT08, BN08} we presented a general mathematical discussion of
these GKS-Lindblad equations. There we gave also references to papers with related results.
This analysis included studies of continuity, of the geometry of paths, details of decay;
all of this related to the decomposition of the Hilbert space into orthogonal subspaces.
Importance of this orthogonal decomposition has been noted, f.e. in \cite{OC10}.

Here we drop the assumption of a continuous time, objects of our study
are also discrete-time semigroups. They are gaining interest as engineering processes
to be used in handling of Quantum Information,
see \cite{KLV00, BG07, SW10} and references therein.
The main results of structuring the Hilbert space are the same for discrete and continuous time;
but here we find special cases of perpetual oscillations and of sudden decay
which may appear only in discrete, not in continuous time.
We offer a new look onto the algebraic essentials
concerning stationary states and invariant operators,
and we highlight the specialities where complete positivity enters.

Our studies concern the interplay, in the course of time evolution, between the sets of states and
the sets of observables with the underlying Hilbert space.
A course grained structure of the Hilbert space emerges, a decomposition into
mutually orthogonal subspaces. There is a precise relation of the state space structure to the processes
of Decay, Decoherence, Dissipation and Dephasing.
Decay transports states from one special subspace to other, remaining subspaces.
Inside each of those there happens Dissipation and Dephasing;
``between them'', appearing in off-diagonal Matrix-blocks of density matrices,
Decoherence takes place.

Each of these processes has often been studied in special models.
What is new, as far as we know, is the analysis of the general case,
where all these processes may happen simultaneously.
Also the process of Oscillation, emerging on another level then in unitary evolution,
seems not to have been discussed earlier in all generality.
The theorems presented in this paper are model-independent and general, do not depend on special assumptions,
except finiteness of the dimension of the Hilbert space.
We get moreover a new result on the general presence of Decoherence,
a new aspect of its nowadays widely accepted universality, \cite{BHS01}.
It sheds a new light on Schr\"{o}dinger's cat-paradox:
In the framework of a quantum-dynamical semigroup, each coherence between dead and alive cat has to decay.
This is discussed at the end under ``Conclusions''.

\section{The setting of quantum dynamical semigroups in physics}\label{setting}

There is a deep conceptual problem in physics: How can one reconcile reversibility of microscopic mechanics
with irreversibility of macroscopic behavior?
Concerning the mathematical aspect, relevant for the present review, is this special part of the question:
How can phenomenological evolution equations, as, \fe for exponential Decay or for Dissipation in form of diffusion,
giving only semigroups, emerge from the
fundamental laws, which form, when integrated, unitary groups?

W. Pauli was the first one who addressed this problem in quantum mechanics and he introduced the master equation \cite{P28}.
He replaced Boltzmann's ``Sto{\ss}zahlansatz'' by a random-phase-approximation,
and he emphasized the necessity for such extra assumptions to pave the way from reversibility to irreversibility.
In the following there were attempts to avoid the assumptions of a repeatedly appearance of random phases.
From those lines of thought two are of relevance for this short review:
Van Hove, \cite{vH55}, demonstrated the importance of considering different timescales.
One for the ``microscopic'' laws, another one for the ``macroscopic'' evolution,
\fe of cells, in a coarse-grained view of the system.
The weakness of the mutual couplings of the cells, indicated by a parameter $\lambda$, leads to a difference
in the two timescales, appearing there as a factor $\lambda^2$.
Another attack on this problem came from S. Nakajima and R. Zwanzig, \cite{N58,Z60},
using a projection operator to part ``relevant'' from ``irrelevant'' parameters,
combined with an assumption of special initial conditions.
In the case when the total system consists of a part which is to be described and an environment
to which this ``open'' part is coupled,
and when the initial state is a product state, this gives a generalized master equation for open systems.
It contains a memory kernel, not present in phenomenological equations. This is often seen as a drawback.
It is overcome by combining the Nakajima-Zwanzig method with a van Hove weak-coupling
limit, or with other limiting procedures.

A large amount of studies on non-equilibrium physics has been done on this basis,
\footnote{
But there are recent developments, allowing for more general frameworks,
as considering non-Markovian dynamics, \cite{S99},
examining couplings of system with environment beyond the weak-coupling limit, \cite{BF01},
and investigating evolutions with correlated initial conditions, \cite{T07}.}
see the reviews \cite{P79,LS78,S80}.
Mathematically self-contained and rigorous studies on the now ``classical''
weak-coupling limit have been performed by E.B. Davies, \cite{D74,D76,D76b}.

Going the route from unitary dynamics for the total system to irreversible ``dissipative'' dynamics
of its open subsystem is still a difficult task.
Often one wants to formulate evolution equations for open systems directly, without explicit
reference to an environment. Here one has to pose the principle question
on the conditions which have to be fulfilled by such equations, in order to be compatible with
a hidden or unknown unitary evolution in the background.
In the case of linear equations - we are not discussing nonlinear equations as appearing \fe in Thomas-Fermi,
Vlasov or Gross-Pitaevskii theory - the property of complete positivity
appears in addition to the necessity of preserving positivity.
This has been pointed out by K. Kraus, \cite{K71}.
It is not only a consequence of the mapping's background in unitary evolution of open system plus environment,
also the possible case when the open system is entangled with a third system,
independent of the environment, demands it, \cite{BF05}.
The general form for a generator of a completely positive semigroup has been
characterized for finite dimensions of the Hilbert space by Gorini, Kossakowski and Sudarshan, \cite{GKS76},
for norm-continuous semigroups in infinite dimensions by G. Lindblad, \cite{L76}.
\footnote{
The problem of characterizing a general form of strongly continuous,
not norm continuous semigroups \cite{D77,A07}, is still not completely solved.}
Working with the ``Lindblad Equation'' has become a standard tool in the theory of open systems,
and many investigations have been done on this basis,
see \cite{AF01,BP02,BF03,AL07} for surveys.

Probably the simplest of the phenomenological laws containing an ``arrow of time'' is the law of
exponential decay of excited states.
Often it is presented as the most simple Lindblad-equation, treating a two-level system.
It has been ``explained'' early in the development of
quantum mechanics in the theory of resonances and using Fermi's Golden Rule.
Revisiting these theories with mathematical rigor
\footnote{
See \cite{JP95,SW98,MS99,BFS99,DJ04,CGH06,LL09,FGS11,DJN11} and references therein,
 - see also the three volumes \cite{AJP06} for surveys.}
brings examples of deriving equations of decay from Hamiltonian dynamics, \cite{JP97},
but it also demonstrates shortcomings of the old theories:
\emph{Exponential decay} is a simplification. Without approximations and limits it holds neither at very short nor at
very long times, unless the spectrum of the Hamiltonian extends from minus to plus infinity.
This has been ``known'' early, has been proven by D.N. Williams, \cite{W71}, and has then been further studied with details,
see \cite{MS77} and references therein.
\emph{Fermi's Golden Rule} may hold in the framework of perturbation theory,
but in some cases it has to be modified, as has been shown in a simple example in \cite{B96}, in generality in \cite{DJN11}.

Aside from ``Decay'' there are two other themes of general interest:
Approach to equilibrium and decoherence.

``Approach to equilibrium'', ``return to equilibrium'' and ``entropy production'' have now replaced the ``H-theorem''
as projects for investigation.
Recent papers in this area are  \cite{BFS00,DJ03,F07}.
In the framework of quantum dynamical semigroups it has been studied in \cite{F78,FR06}.
Regarding the mathematical aspect of approach to a uniquely defined state of the open system,
it can be considered together with the existence of a non-equilibrium steady state,
which has been proven in a model in \cite{JP02}.

``Decoherence'' is a hot theme. It has to do with the transition from the quantum to the classical,
and it is extensively studied, see \fe \cite{Z03,S04}. A satisfactory theory
seems not to have been developed up to now, according to the meaning stated f.e. in \cite{CSU11}.
Early mathematical studies on this theme date back to the seventies.
Especially H.Primas noted the necessity to regard already large molecules as open systems
in interaction with and entangled with the environment, \cite{P75}.
Modern rigorous works, relevant in the context of open-system dynamics are
\cite{BO03}, the chapter in \cite{AJP06} vol.III by the same authors, \cite{MSB07,O08} and the above mentioned \cite{CSU11}.

Each of these three themes appears in this paper, which is a study on processes in the framework
of dissipative semigroups of evolution, without asking any more on their background.
Approach to a Stationary State is given as Decay, Dissipation and Dephasing.
Together with Decoherence these processes emerge generally as special mathematical aspects of the structuring
of the Hilbert space, related to the quantum dynamical semigroup.
In the framework of general open-system dynamics no reference to a perturbed Hamiltonian is necessary.
Decaying states may be either mixed or pure. Stationary states may be present in multitude.
Their existence is related to a splitting of the Hilbert space into orthogonal subspaces.
Decoherence between subspaces turns out to be a necessity,
unless there is a unitary dynamical equivalence between them.
A new aspect appears, which, to our knowledge, has not yet been studied in this framework:
It is Oscillation. Here it makes a difference whether the evolution is continuous
in time, or a discrete ``bang-bang'' process.
Both cases are completely classified.

\section{Setup and results}\label{setup}

Consider an $N$-dimensional Hilbert space $\Hh$, equipped with an algebra $\A$
of observables -- representable as self-adjoint $N\times N$ matrices --, and states --
representable as positive  $N\times N$ matrices with trace equal to one.

Time evolution of states appears in the form of
maps $\TT^t$, $t\geq 0$, either $t\in \R_+$ or $t\in \Z_+$,
which have the properties of forming a semigroup
and being compatible with mixing and  with convex decomposition:
\beqa
\ \TT^{t+s}(\rho)&=&\TT^t(\TT^s(\rho)),\label{semigroup}
\\ \TT^t(\lambda\rho+(1-\lambda)\sigma)&=&\lambda\TT^t(\rho)+(1-\lambda)\TT^t(\sigma).\label{mixing}
\eeqa
We define extended operators $\TT^t$ which act
on the whole vector-space of $N\times N$ matrices.
Condition (\ref{mixing}) is extended to give linearity;
the original definition of mapping states onto states appears for the extended operators
as the third and the fourth item in the following:
\begin{define}\label{defTT}
{\bf Properties of the time evolution in the Schr\"{o}dinger picture:}
The maps $\TT^t$ have the properties that they
\begin{enumerate}
  \item  form a semigroup, $\TT^{t+s}(\rho)=\TT^t(\TT^s(\rho)),$
  \item  can be extended as linear super-operators, acting on $N\times N$ matrices,
  \item  preserve positivity, $\rho\geq 0 \rightarrow \TT^t(\rho)\geq 0,$
  \item  preserve normalization, $\Tr(\TT^t(\rho))=\Tr(\rho),$
  \item  are completely positive.
\end{enumerate}

\end{define}
\medskip
The extension is, using $r_j\in\R_+$:
$$\TT^t (r_1\rho_1 + i\,r_2\rho_2 - r_3\rho_3 - i\,r_4\rho_4) :=
\TT^t (r_1\rho_1) + i\,r_2\TT^t (\rho_2) - r_3\TT^t (\rho_3) - i\,r_4\TT^t (\rho_4),$$

In the Heisenberg picture there are the dual maps $\TT^{t\dag}$ acting on observables.
They too can be extended as maps acting on the vector-space of matrices.
The properties are almost the same as for the $\TT^t$, only the preservation of the trace has to be replaced by

\medskip
 \textit{4$^\dag .$}\,\, \textit{ the maps }  $\TT^{t\dag}$ \textit{ preserve the unity, they are unital,} $\TT^{t\dag} (\one)=\one.$

\medskip

The properties (2), (3), (4), (5) of  $\TT^t$ are necessary and sufficient for its dilation,
i.e. extension to a unitary map of a larger system, \cite{S55,C75}.
The condition of complete positivity is therefore essential for maps representing quantum dynamics
(see \cite{K71}). Moreover, it makes a fundamental difference to merely positive maps when the system under study is
entangled with other systems, (see \cite{BF05}).
Also inside the system  it enters in special details concerning the entanglement of subsystems.
To highlight these specialities we treat  complete positivity separately in  Section \ref{Sectioncohcp},
namely in Proposition \ref{cohsplit}, in the following Proposition Lemma and Theorem,
in proving Theorem \ref{coherence}, and then Theorem \ref{splittingxc},
which gives a most detailed characterization of the splitting properties of the Hilbert space.
It will not be used before that, in the derivation of
Theorems (\ref{splitting}), (\ref{splitting2}), (\ref{splitting3}):

\begin{thm} \textbf{Splitting of the Hilbert space into two subspaces:}\label{splitting}                        

\begin{enumerate}
  \item The set \beq \label{Dspace1}
  \D:=\{\psi: \quad \forall \,\rho\quad
  \langle\psi|\TT^t (\rho)|\psi\rangle\rightarrow 0 \quad\rm{as}\quad t\rightarrow\infty\}
  \eeq
  is a subspace of $\Hh$. It is the largest subspace
  such that for each normalized density matrix $\rho$
\beq\label{characterization}
\lim_{t\rightarrow\infty}\Tr(P_\D\TT^t(\rho))=0,\qquad {\rm and}\qquad
\lim_{t\rightarrow\infty}\Tr(P_\Rr\TT^t(\rho))=1,
\eeq
where $P_\D$ denotes the projector onto $\D$, and  $P_\Rr$ the projector onto $\Rr =\D^\perp$.
$\D$ may be trivial, containing nothing but the vector zero, but it is in any case not the whole Hilbert space $\Hh$.
  \item The space $\Rr$, the orthogonal complement to $\D$, has the property
  \beq
  \forall \psi\in\Rr\quad\exists\,\rho:\quad \langle\psi|\TT^t (\rho)|\psi\rangle=const.>0.
  \eeq
  \item The Decay in Equation (\ref{characterization}) is monotonous. \label{Dspace3}
\end{enumerate}
\end{thm}
In the Heisenberg picture the characterizations are:
\beq
\D=\{\psi: \quad \TT^{t\dag}(|\psi\rangle\langle\psi|)\rightarrow 0 \quad\rm{as}\quad t\rightarrow\infty\}\label{Dspace};
\eeq
\beq\label{HeisenbergD}
\D \textrm{ is the maximal subspace } \V \textrm{ for which }\quad \lim_{t\rightarrow\infty}(\TT^{t\dag}(P_\V))=0,
\eeq
\beq\label{HeisenbergR}
\Rr \textrm{ is the minimal subspace } \V \textrm{ for which }\quad \lim_{t\rightarrow\infty}(\TT^{t\dag}(P_\V))=\one.
\eeq
Loosely speaking, $\D$ contains all the decaying states, while no state in $\Rr$ is decaying completely,
and no part of any state leaves $\Rr$.
The proof is presented in parts, in Subsection \ref{details1}.

In the following we talk about
the \textbf{range} of a state $\rho$, meaning the range of the density matrix which represents this state.
We note that in the context of quantum systems on spaces with infinite dimension, \fe in \cite{H01,FR06},
instead of ``range'' the naming ``support'' is used, generalizing the notion ``support of a probability measure''
from commutative to non-commutative theories.

The investigations on the structure of $\Hh$ are investigations on the ranges
of states in the course of time.
The subspace $\D$ contains the ranges of all decaying states.
The subspace $\Rr$, containing the ranges of all stationary states,
can further be decomposed according to the structure of the set of stationary states.
This structure is related to the set-structure of the ranges.
We call the range of a stationary state  a \textbf{minimal stationary range} if it
does not contain a smaller range of a stationary state.
Again there are orthogonality relations and completeness, giving the second main theorem,
proven in Subsection \ref{details2}:

\begin{thm}\textbf{Splitting of the non-decaying subspace:}\label{splitting2}   

The space  $\Rr$ can be decomposed into a
direct sum of mutually orthogonal subspaces,
where each one of them is a minimal stationary range.
\end{thm}

The conditions on $\TT^t$ listed in Definition \ref{defTT} do not exclude
unitary evolutions, and one can ask about the form taken by the definitions and theorems in these
well known special cases.
For unitary evolutions there is no decay, and $\Rr$ is the whole Hilbert space.
The unitary operators form a group, generated by a Hamiltonian.
(In case of a discrete group, the Hamiltonian is uniquely defined, if
one restricts its spectrum to a certain finite interval.)
Those density matrices which can be diagonalized simultaneously with the Hamiltonian
give the stationary states.
A minimal stationary range is a one-dimensional space containing an eigenvector
of the Hamiltonian, and the decomposition of $\Rr$ is the spectral decomposition.
It is unique, unless there appears a degeneracy in the spectrum of the Hamiltonian.
Now, for general semigroups $\TT^t$, there is an analogous situation -
the proof is in Subsection \ref{details3}:

\begin{thm}\textbf{Equivalence of splittings:}\label{splitting3}   

If there are several possibilities to decompose $\Rr$ into mutually orthogonal minimal stationary ranges,
then all these decompositions, together with the invariant states, are unitarily equivalent
under unitary transformations commuting with the time evolution restricted to
states with range in $\Rr$.
\end{thm}

Such unitary transformations can be regarded as expressing symmetries of the dynamics;
if not existing on the whole space then at least on the subspace $\Rr$.
The ``symmetry''-group they form
can be analyzed generally in all detail, if complete positivity of $\TT$ is taken into account.
These studies are presented in the Subsection \ref{cohblocks}, where
a closer look at coherence between stationary states is taken.

``Coherence'' generalizes here the standard notion
``coherent superposition'', which means forming the Hilbert space vector
$|\chi\rangle=\alpha|\psi\rangle+\beta|\phi\rangle$
as opposed to incoherent mixing, forming the density matrix
$\rho=|\alpha|^2|\psi\rangle\langle\psi|+|\beta|^2|\phi\rangle\langle\phi|$.
The density matrix $|\chi\rangle\langle\chi|$ has the same diagonal matrix elements as $\rho$,
it shows the coherence in the off-diagonal
matrix elements $\alpha\beta^\ast|\psi\rangle\langle\phi|$ and $\alpha^\ast\beta|\phi\rangle\langle\psi|$.
This is generalized in the following way:
If $\V=\range(\rho)$ and  $\W=\range(\sigma)$ are mutually orthogonal subspaces of $\Hh$,
a \textbf{coherence} between $\rho$ and $\sigma$ is formed by matrix elements
$|\psi\rangle\langle\phi|$  and $|\phi\rangle\langle\psi|$ connecting these subspaces,
with $|\psi\rangle\in\V$ and $|\phi\rangle\in\W$.
Again there is a possible generalization of the analysis of unitary evolutions using matrices
to an analysis of semigroups using block-matrices:

\begin{define} \textbf{Matrix-blocks:} \label{matrixblocks}                                         
Given two orthogonal subspaces $\V$ and  $\W$ of $\Hh$, we consider three sections of the linear space of matrices:
\begin{itemize}
  \item $\M_{\V}$ is the subspace of matrices spanned by the $Q_\psi :=|\psi\rangle\langle\psi|$ with $\psi\in\V$,
  \item $\M_\W$ is spanned by the $Q_\phi$ with $\phi\in\W$,
  \item the coherence-section $\M_C$ consists of the off-diagonal blocks:
$$\M_C = \lspan (\{|\psi\rangle\langle\phi|, |\phi\rangle\langle\psi|;\,\, \psi\in\V,\, \phi\in\W\})
=\M_{C,1}\oplus\M_{C,2}.$$
\end{itemize}
\end{define}
This can be visualized, representing a matrix $M$ as block-matrix:
$$ M =
\left(
  \begin{array}{ccc}
    \cdots  & \cdots & \cdots \\
    \cdots & M_\V & M_{C,1} \\
    \cdots & M_{C,2} & M_\W \\
  \end{array}
\right)
\qquad \verb"acting on" \qquad \Hh =
\left(
  \begin{array}{c}
    \cdots \\
    \V \\
    \W \\
  \end{array}
\right)
$$

Generally, $\TT$ maps $\M_\D$ into the full space of matrices,
it maps $\M_{C(\D ,\Rr)}$ into $\M_{C(\D ,\Rr)}\oplus\M_\Rr$, and $\M_\Rr$ into itself;
thats the monotonicity in the process of Decay, proven in Subsection \ref{details1}.
The mapping of $\M_\Rr$ can be decomposed into mappings of matrix blocks $\M_\V$ onto itself
and of each $\M_{C(\V , \W)}$ onto itself, where $\V$ and $\W$ are mutually orthogonal minimal stationary ranges.
Notice the different behavior of such a $\M_{C(\V,\W)}$  and $\M_{C(\D,\Rr)}$;
it is demonstrated in an example in \cite{BN08}.
Dissipation goes on in each $\M_\V$ separately, leading to its unique stationary state.
Also Dephasing takes place inside each $\M_\V$, annihilating all off-diagonal
matrix elements, when a basis is used, in which the stationary state is diagonal.
In the Heisenberg picture, $\M_\Rr$ is mapped into the full space of matrices;
this is inconvenient in studies on approach to a stationary state.
Therefore it is convenient to ignore $\D$, unless Decay is studied, and to use a cut off evolution
of observables which has the same nice decomposition into mappings of matrix blocks.
The definition is at the end of the Subsection \ref{details1},
the possibility for decomposition is shown in  Subsection \ref{details2}.

Coherence appears in the off-diagonal matrix-blocks.
Only symmetries between ranges of stationary states can avoid their decay:

\begin{thm}\textbf{Decay vs. stationarity vs. oscillation of coherence:}\label{coherence}   

\begin{enumerate}
   \item Evolution of coherence between stationary ranges can be decomposed into mutually independent
   evolutions of coherence between the minimal stationary ranges which give a decomposition of $\Rr$.          \label{coherence1}
   \item A stationary coherence between minimal stationary ranges $\V$ and $\W$ exists
    if and only if the decomposition of $\V\oplus\W$ into minimal stationary ranges
    is not unique. This can happen only if
    the dynamics of the subsystems are unitarily equivalent
        under unitary transformations commuting with the time evolution restricted to $\Rr$.          \label{coherence2}
   \item For continuous time evolution oscillation of coherence between stationary ranges $\V$ and $\W$ may occur
   only if the dynamics of the subsystems are unitarily equivalent in a weaker sense,
        namely restricted to the blocks on the diagonal:                            \label{coherence3}
\beq
\exists\, U:\,\V\rightarrow\W,\quad U^\dag=U^{-1},\qquad \forall\sigma=
P_\V\tdt\sigma\tdt P_\V:\quad \TT^t(\sigma)=\TT^t(U\tdt\sigma\tdt U^\dag).
\eeq
\end{enumerate}
\end{thm}

The proof is presented at the end of Subsection \ref{cohblocks}.
Complete positivity is not needed to prove the first and second items,
it is needed only to prove the third item on oscillations,
and it enables a sharper formulation of item 1.

There are no other possibilities of perpetual oscillations for evolutions in continuous time.
The situation is slightly different for discrete dynamical semigroups.
Again perpetual oscillations of coherence can occur only
if there exists a dynamical unitary equivalence between subsystems,
but there may occur also perpetual oscillations of states inside minimal stationary ranges, and their existence may be
accompanied by  oscillations of coherence.
These specialities appearing in discrete time are presented in the Subsection \ref{discretesection}.

Using complete positivity of $\TT$ it is possible to write down the general form of a stationary state,
the analogue to the diagonalization by using the spectral theorem in finite dimensions:

\begin{thm}\textbf{Complete characterization of the decomposition of $\Rr$:}\label{splittingxc}
There is a unique decomposition of the subspace $\Rr$ as
\beq
\Rr=\bigoplus_k \U_k \oplus\bigoplus_\ell \X_\ell,
\eeq
where each $\U_k$ is the range of one and only one stationary state $\rho_k$, and each $\X_\ell$
can be further decomposed into minimal ranges of stationary states as
\beq\label{statrhoform}
\X_\ell=\bigoplus_\alpha \V_{\ell,\alpha}\simeq \C^{m(\ell)}\otimes \V_\ell,\qquad \V_\ell\simeq\V_{\ell,\alpha}\quad\forall \alpha,
\eeq
so that each stationary state can be written as
\beq
\sum_k \lambda_k\cdot\rho_k +\sum_\ell \mu_\ell \cdot\sigma_\ell\otimes\tau_\ell,
\eeq
where $\lambda_k$ and $\mu_\ell$ are non-negative numbers, $\sum_k\lambda_k +\sum_\ell \mu_\ell=1$,
$\sigma_\ell$ can be any positive matrix on $\C^{m(\ell)}$ with $\Tr\sigma_\ell=1$,
and $\tau_\ell$ is a unique special normalized density matrix on $\V_\ell$.

There is no stationary coherence between the $U_k$, between the $\X_\ell$
or between a $\U_k$ and an $\X_\ell$.
\end{thm}
The proof uses complete positivity and is at the end of subsection \ref{cohblocks}.

In the following parts of this paper we will, with slight abuse of formulation,
write all the formulas related to splitting a subspace $\X$ as
in equation (\ref{statrhoform}) with ``$=$'' instead of ``$\simeq$''.

\section{Details and proofs concerning stationary states}\label{details}

\subsection{Splitting into a decaying subspace and enclosures}\label{details1}

\begin{proof}{\textit{of Theorem \ref{splitting}.1:}}

Consider the positive not normalized operators
$Q_\psi = | \psi \rangle \langle \psi |$ and
consider vectors $\psi$ and $\phi$ for which
$\TT^{t\dag}(Q_\psi)\rightarrow 0$ and  $\TT^{t\dag}(Q_\phi)\rightarrow 0$ as $t\rightarrow\infty$.
Preservation of positivity implies preservation of operator-inequalities.
Now the inequality
\beq\label{qpsi}
0\leq Q_{\psi+\phi}\leq Q_{\psi+\phi}+Q_{\psi-\phi}=2Q_\psi+2Q_\phi
\eeq
implies $\TT^{t\dag}(Q_{\psi+\phi})\rightarrow 0$ as $t\rightarrow\infty$.
So the set $\D$, as characterized in Equ. (\ref{Dspace}),
is a subspace of $\Hh$.
Decomposing $P_\D=Q_\psi +Q_\phi +\ldots$ shows that $\D$
obeys (\ref{HeisenbergD}). By duality
$\Tr(P_\D \TT^t(\rho))\rightarrow 0$
for each $\rho$.

Since $\Tr(\TT^t(\rho))=const$, the projector $P_\D$ is strictly smaller than $\one$.
The evolution $\TT^{t\dag}$ is unital, and for $\Rr=\D^\perp$
$$\TT^{t\dag}(P_\Rr) = \TT^{t\dag}(\one - P_\D) = \one - \TT^{t\dag}(P_\D)  \rightarrow\one$$
gives (\ref{HeisenbergR}).
By duality $\Tr(P_\Rr \TT^t(\rho))\rightarrow 1$ for each $\rho$.
\end{proof}

Not every state of a finite dimensional system can decay. Counterparts
to ranges of decaying states are ranges which can not grow in time,
and ranges of stationary states. These have to be orthogonal to $\D$, by definition in equ. (\ref{Dspace1}).

\begin{define}\textbf{Enclosure}\label{defencl}                                                   

In case that for each state $\sigma$ with range in $\V$
also the range  of $\TT^t(\sigma)$ is contained in $\V$,
for each following time $t$, then this subspace $\V$ is called an \textbf{enclosure}.
\end{define}

\begin{proposition}\textbf{Stationary ranges are enclosures.}\label{enclosure}                             

If $\V$ is the range of a stationary state, then it is an \textbf{enclosure}.
\end{proposition}
\begin{proof}
Let $\rho$ be a stationary state which has $\V$ as its range,
consider another state $\sigma$, the range of which is contained in $\V$. Then
$\exists\, \epsilon>0$ such that $\epsilon\sigma<\rho$.
By preserving positivity, also the ordering is preserved in the course of time.
So $\epsilon \TT^t(\sigma)<\TT^t(\rho)=\rho$, which implies that the range of $\TT^t(\sigma)$
is still contained in $\V$.
\end{proof}

Stationary states and invariant operators can be constructed.
From the general theory of coupled linear differential equations --
and of coupled difference-equations for discrete time as well -- we know
that any solution is a finite linear combination of such functions which are
products of a polynomial with an exponential function.
For evolution equations of matrices in time this concerns the matrix elements,
i.e. the functions $f_{k,\ell }(t)=\langle k|\TT^t(\rho)|\ell\rangle$.
These functions can not increase indefinitely as $t\rightarrow\infty$, since the trace norm of positive
matrices stays constant and their matrix elements are bounded by
$|\langle k|\rho|\ell\rangle|\leq \,\|\rho\| \leq \Tr(\rho)$.
So we can define  \textbf{the mean of states and of operators}:                                                  

For continuous time-evolution:
$$\bar{\rho} := \lim_{T\rightarrow\infty}\frac{1}{T}\int_0^T \TT^t(\rho)dt,  \qquad
\bar{A} := \lim_{T\rightarrow\infty}\frac{1}{T}\int_0^T \TT^{t\dag}(A)dt.$$

In case the semigroup is discrete, defined for times $t=n\tau,\,n \in \N$:
$$\bar{\rho} := \lim_{N\rightarrow\infty}\frac{1}{N+1}\sum_{n=0}^N \TT^{n\tau}(\rho),  \qquad
\bar{A} := \lim_{N\rightarrow\infty}\frac{1}{N+1}\sum_{n=0}^N \TT^{n\tau\dag}(A).$$

$\TT^{t\dag}$ preserves positivity, and hence, for positive matrices $A$, it preserves the inequality
$0\leq A\leq\,\|A\|\, \cdot\one$.
So, the norm of $\TT^{t\dag} (A)$ is not increasing in time and $\|\bar{A}\|\leq \|A\|$.
Also the trace norm of $\rho$ is not increasing.
By dominated convergence, taking the trace of $\rho\cdot A$ can be exchanged with these limiting processes.
Duality gives thus $\Tr(\bar{A}\rho) = \Tr(A\bar{\rho})$.
Considering also $\bar{\rho}$ in place of $\rho$ and using invariance in time of $\bar{\rho}$,
implying $\overline{\bar{\rho}}=\bar{\rho}$, gives thus
\beq\label{stat}
\forall \rho,\, \forall A \qquad \Tr(\bar{A}\rho) = \Tr(A\bar{\rho}) = \Tr(\bar{A}\bar\rho).
\eeq

Considering again the matrix elements as functions of time, one can note
that they have the special form: Constant plus oscillation plus an exponentially decreasing part.
A positive matrix stays positive, so oscillations have to be bounded by a positive constant,
and the mean of a positive operator $A$ is zero if and only if its limit as $t\rightarrow\infty$ is zero.

With these preliminaries one can state and prove

\begin{proposition}\textbf{The subspace $\Rr$ supporting stationary states.} \label{rspace}

The subspace $\Rr$ has the properties           
\begin{enumerate}
  \item $\forall\, \rho:\quad \rm{range}(\bar{\rho})\subset\Rr$.
  \item $\exists\, \bar{\rho}:\quad \rm{range}(\bar{\rho})=\Rr$.\label{rspace2}
  \item $\Rr$ is an enclosure: \quad
  $\forall \rho$ with ${\rm range}(\rho)\subset\Rr$ $\Rightarrow$ $\forall t:\,
   {\rm range}(\TT^t(\rho))\subset\Rr$.\label{rspace3}

\end{enumerate}
\end{proposition}
\begin{proof}

The range of a stationary state has to be orthogonal to $\D$, by the defining equation (\ref{Dspace1});
this implies \ref{rspace}.1.

On the other hand,
there exists no $\psi\in\Rr$ with its positive operator $Q_\psi=|\psi\rangle\langle\psi|$ decaying to zero,
by definition of $\D$ and $\Rr$.
So, as stated above, for each $\psi\in\Rr$, the mean $\overline{Q_\psi}$ is positive and not zero.
Moreover, splitting into matrix-blocks, $\overline{Q_\psi}=(\overline{Q_\psi})_\D +(\overline{Q_\psi})_C+(\overline{Q_\psi})_\Rr$,
followed by again forming the mean,
using $(\overline{Q_\psi})_\D\leq P_\D$, since the norm is not increased, implying
$$
\overline{(\overline{Q_\psi})_\D}\leq \overline{P_\D}=0,
$$
and
$$0\neq \overline{Q_\psi}=\overline{(\overline{Q_\psi})}=\overline{(\overline{Q_\psi})_\D} +\overline{(\overline{Q_\psi})_C}+\overline{(\overline{Q_\psi})_\Rr} \geq 0,$$
implies $(\overline{Q_\psi})_\Rr \neq 0$.
Using any $\rho$ with $\range(\rho)=\Rr$,
the equation (\ref{stat}) gives therefore for each $\psi\in\Rr$
\beq
\langle\psi|\bar{\rho}|\psi\rangle =\Tr(\bar{\rho}\cdot Q_\psi)=\Tr(\rho\cdot\overline{Q_\psi})
=\Tr(P_\Rr\rho P_\Rr\cdot\overline{Q_\psi})=\Tr(\rho\cdot(\overline{Q_\psi})_\Rr)\neq 0,
\eeq
so $\range(\bar{\rho})=\Rr$, as stated in \ref{rspace}.\ref{rspace2}.

\ref{rspace}.\ref{rspace3} follows then from  Proposition (\ref{enclosure}).
\end{proof}
\begin{proof}{\textit{of Theorem \ref{splitting}.2:}}

 Take the $\bar{\rho}$ of Proposition \ref{rspace}.\ref{rspace2}.
\end{proof}
The study of mappings of matrix blocks begins with
\begin{lem}\label{monotone}
Let $\V$ be an enclosure, $\W=\V^\perp$, and $\M_C$ the coherence section
for $\M_\V$ and $\M_\W$.
The evolutions $\TT^t$ do not map $\M_C$ into $\M_\W$.
\end{lem}
\begin{proof}
$\V$ is an enclosure, so $\TT^t$ maps $\M_{\V}$  into  $\M_{\V}$.
Now consider a density matrix $\rho$, decompose it
according to the sections as $\rho=\rho_\V +\rho_C +\rho_\W$,
and consider the action of $\TT^t$ onto the
matrices $\kappa_{\eps +}$ and $\kappa_{\eps -}$, where $\kappa_{\eps \pm} = \frac 1\eps \rho_\V \pm\rho_C +\eps \rho_W$,
with $\eps >0$.
These matrices are positive, which can be seen by
$$\langle\chi|\kappa_{\eps \pm}|\chi\rangle = \langle\chi_{\eps \pm}|\rho|\chi_{\eps \pm}\rangle ,$$
with a general vector $\chi$, decomposed as $\chi=\psi+\phi$,\, $\psi\in\V,\,\phi\in\W$,\,
and with $\chi_{\eps\pm}=\psi/\sqrt{\eps}\pm\sqrt{\eps} \phi$.

The map $\TT^t$ conserves positivity, so $P_\W\TT^t(\kappa_{\eps \pm})P_\W \geq 0$.
Since $\V$ is an enclosure, this part can not consist of more than
$\eps P_\W\TT^t(\rho_\W)P_\W \,\pm\, P_\W\TT^t(\rho_C)P_\W$,
and this can only then be positive for all $\eps$, if
$P_\W\TT^t(\rho_C)P_\W=0$.
So, since the set of all $\rho_C$ span all $\M_C$, $\TT$ does not map the coherence section $\M_C$ into $\M_W$:
\beq\label{cohencl}
\TT^t(\M_C)\cap\M_\W =\{0\}
\eeq
\end{proof}
\begin{proof}{\textit{of Theorem \ref{splitting}.3:}}

Decompose $\sigma=\TT^s(\rho)$ as $\sigma=\sigma_\Rr +\sigma_C +\sigma_\D$.
 $\Rr$ is an enclosure (Proposition \ref{rspace}.\ref{rspace3}),
 so $\TT^t$ maps neither $\sigma_\Rr$ nor $\sigma_C$ into $\M_\D$, and
 \beqa
 \Tr(P_\D \TT^{t+s}(\rho))&=&\Tr(P_\D \TT^t(\sigma))=\Tr(P_\D \TT^t(\sigma_\D))\nonumber
 \\ &\leq &\Tr(\TT^t(\sigma_\D))=\Tr(\sigma_\D)=\Tr(P_\D \TT^s(\rho)).\nonumber
 \eeqa

\end{proof}
As a consequence, $\D$ may be ignored when considering the stationary states.
And, when  using invariant operators, we are in these investigations not interested in their actions
on the space $\D$, the space of decaying states.
We study only the operators $A=P_\Rr\cdot A\cdot P_\Rr$ and define
a cut off evolution $\S$ of matrices, mapping $\M_\Rr \rightarrow \M_\Rr$:
\begin{define} \textbf{The cut off evolution of observables.}                                                          

For observables $A=P_\Rr\tdt A\tdt   P_\Rr $ we set
\beq\S^t(A):=P_\Rr \tdt \TT^{t\dag}(A) \tdt P_\Rr
\eeq
\end{define}
Since $\Rr$ is an enclosure, we have:
\begin{lem}                                                               
The cut off evolutions $\S^t$ form a semigroup.
\end{lem}
\begin{proof}
Let $\rho_\Rr=P_\Rr\rho P_\Rr$, for any $\rho$, consider  $A=P_\Rr\tdt A\tdt   P_\Rr $. Then
$$\forall\rho\qquad \Tr(\rho \cdot  \S^t(\S^u(A)))=\Tr(\rho_\Rr \cdot  \TT^{t\dag}(\S^u (A)))=
\Tr(\TT^t(\rho_\Rr)\tdt  \S^u (A))$$
$$=\Tr(P_\Rr\TT^t(\rho_\Rr)P_\Rr\tdt  \S^u (A))=\Tr(\TT^{t+u}(\rho_\Rr)\cdot A)=
\Tr(\rho_\Rr \cdot (\TT^{(t+u)})^\dag(A))$$
$$=\\Tr(\rho\cdot \S^{(t+u)} (A))$$
\end{proof}

\subsection{Orthogonal decompositions of $\Rr$}\label{details2}

It is not at all obvious a priori that orthogonality of subspaces plays a role in dissipative
time evolution, where orthogonality is not a generally conserved property.
We have the conservation of positivity by $\TT$, by $\TT^\dag$ and also by $\S$ instead.
The necessary information on conservation of orthogonality of special subspaces is gained by studying the matrices in block-form,
as in Definition (\ref{matrixblocks}).
An essential tool are the decompositions of $\TT$ and of $\TT^\dag$
into their actions onto the matrix blocks:
\begin{proposition} \textbf{Mappings of matrix-blocks.} \label{blockindep}                         

Consider a subspace $\V$ which is an enclosure. Then
also $\W = \V^\perp\cap\Rr$ is an enclosure.
Both  $\TT^t$ and $\S^t$ (not $\TT^{t\dag}$) map $\M_{\V}$ into $\M_{\V}$, $\M_{\W}$ into $\M_{\W}$ and $\M_C$ into $\M_C$.
\end{proposition}

\begin{proof}

Consider a stationary state $\bar{\rho}$ which has the whole subspace $\Rr$ as range;
its existence is guaranteed by Proposition (\ref{rspace}.2).
Decompose it  according to the sections: $\bar{\rho}=\bar{\rho}_\V +\bar{\rho}_C+\bar{\rho}_\W$,
note that the range of $\bar{\rho}_\W$ is the whole subspace $\W$.
Since $\bar\rho$ is stationary, and since $P_\W \TT^t(\bar\rho_{\V})=0$, and also $P_\W \TT^t(\bar\rho_C)P_\W=0$
by Lemma \ref{monotone},
one has $P_\W\TT^t(\bar\rho _\W)P_\W=\bar\rho _\W$;
nothing enters the subspace $\W$.
It is also true, that nothing of $\TT^t(\bar{\rho}_\W)$ is outside of $\M_\W$,
because,
starting the evolution with $\bar{\rho}_\W$ alone, without any state in the complement of $\W$,
this would contradict preserving the trace and positivity.
So $\bar{\rho}_\W$ is a stationary state.
$\W$ as its range is therefore an enclosure.
Then in analogy to (\ref{cohencl}): $\TT^t(\M_C)\cap\M_\V=\{0\}$.

Now, this mapping of matrix blocks into itself can be formalized as:
$$
\forall\rho \quad{\rm with }\,\, \range(\rho)\subset\Rr,\quad  \TT^t(P_\alpha\rho P_\beta)=P_\alpha \TT^t(\rho))P_\beta,
$$
where both $P_\alpha$ and $P_\beta$ may be either $P_\V$ or $P_\W$.
Consider $A=P_\Rr\tdt A\tdt P_\Rr$ and any $\rho$ with $\range(\rho)\subset\Rr$:
$$
\Tr(\rho\tdt \S^t(P_\beta A P_\alpha))=\Tr(\TT^t(\rho)\tdt P_\beta A P_\alpha)
=\Tr(\TT^t( P_\alpha\rho P_\beta)\tdt A)=\Tr(\rho\tdt P_\beta \S^t(A) P_\alpha).
$$
This implies $\S^t(P_\beta A P_\alpha)= P_\beta \S^t(A) P_\alpha$.
\end{proof}

Conservation of self-adjointness
comes from linearly extending (following Definition \ref{defTT}) the maps $\TT^t$,
first defined as mapping state to state, to mappings of general matrices, where
obviously $\TT^t(\sigma^\dag)=[\TT^t(\sigma)]^\dag$.
When using complete positivity of the maps it makes sense to consider moreover the two
parts of decoherence blocks $\M_C$, one above and the other below the diagonal, separately.
This will be done in Subsection \ref{cohblocks}.
Such a decomposition is not needed before that.

Stationary states form a convex set in a vector space with finite dimensions.
There are extremal elements of this set, \textbf{extremal stationary states}.
There are one-to-one relations between extremal stationary states,
minimal stationary ranges and minimal enclosures:
\begin{proposition}\textbf{Equivalences and Decomposition}\label{equi}                         
\begin{enumerate}
  \item Each enclosure inside of $\Rr$ is the range of a stationary state.   \label{equi1}
  \item Different extremal stationary states have different ranges.   \label{equi2a}
  \item The range of an extremal stationary state is a minimal enclosure inside of $\Rr$ and vice versa. \label{equi2}
  \item A minimal enclosure in $\Rr$ is a minimal stationary range and vice versa.   \label{equi3}
  \item Each enclosure inside of $\Rr$, especially $\Rr$ itself, can be decomposed as an orthogonal sum of minimal enclosures. \label{equi4}
\end{enumerate}
\end{proposition}
\begin{proof}

(\ref{equi1})
Consider again a stationary state $\bar{\rho}$ which has the whole subspace $\Rr$ as range
and decompose it into matrix blocks $\bar{\rho}_\V+\bar{\rho}_\W+\bar{\rho}_C$,
where $\V$ is an enclosure inside $\Rr$.
By Proposition \ref{blockindep} evolutions of these parts are mutually independent, so
$\overline{\bar{\rho}_\V}=\bar{\rho}_\V$ is a stationary state and $\V$ is its range.

(\ref{equi2a})
Let $\rho$ be an extremal stationary state and $\V$ its range.
Suppose there is  another stationary state $\sigma$ with $\range(\sigma)=\V$.
Since the dimension of the space is finite, there would be an $\eps >0$ such that $\rho -\eps\sigma \geq 0$.
But then $\tau=(\rho-\eps\sigma)/(1-\eps)$ were also a stationary state
and $\rho$ could be decomposed into $\sigma$ and $\tau$, in contradiction to its extremality.

(\ref{equi2})
Let again $\rho$ be an extremal stationary state and $\V$ its range. This $\V$ is an enclosure (Proposition \ref{enclosure}).
Suppose there is a smaller enclosure inside: it would contain the range of another stationary state $\sigma$
(item \ref{equi1}). This would enable a decomposition of $\rho$ into $\sigma$ and $\tau$
as in the proof of item \ref{equi2a}, contradicting the extremality.
On the other hand,
a minimal enclosure  $\V\subset\Rr$ is the range of a stationary state $\rho$ (item \ref{equi1}).
If $\rho$ could be decomposed as $\rho=\alpha\sigma+(1-\alpha)\sigma'$, (necessarily with $\range(\sigma)\subset\range(\rho)=\V$),
there were a maximal $\eps$ of those numbers such that $\rho-\eps\sigma >0$.
The $\tau$ constructed as above, but using this $\eps_{max}$, were then at the border of the set of states with range in $\V$.
It would have an eigenvector inside of $\V$ to the eigenvalue zero, so its range,
which is again an enclosure, were strictly smaller than $\V$,
a contradiction to the minimality of $\V$.

(\ref{equi3}) Since ``stationary range'' is defined as the range of a stationary state,
the equivalence between extremal states and minimal enclosures, shown in proving (\ref{equi2}), gives the
equivalence between minimal stationary ranges and minimal enclosures.

(\ref{equi4}) Now consider some enclosure $\V\subset\Rr$.
If it is minimal, there is nothing to prove. Otherwise it contains a minimal enclosure $\V_1$.
The subspace $\W_1=\V_1^\perp\cap\Rr$ is also an enclosure (Proposition \ref{blockindep}),
and the intersection $\X_2=\W_1\cap\V$ of two enclosures is also an enclosure.
Either $\X_2$ is minimal, in which case the proof is finished, or it contains a minimal enclosure $\V_2$.
Then we proceed with $\W_2=(\V_1\cup\V_2)^\perp\cap\Rr$,  and proceed further until the decomposition of $\V$ is complete.
\end{proof}

Considering operators under cut off evolution enlarges the list of equivalences:
\begin{proposition} \textbf{Equivalences with invariant operators} \label{sinv}                          

\begin{itemize}
  \item A projector $P_\V \neq 0$ is invariant under $\S$, iff it projects onto an enclosure $\V\subset\Rr$.
  \item An observable $A=P_\Rr\cdot A\cdot P_\Rr$ is invariant under $\S$, iff its spectral projections are invariant.
  \item $\V$ is a minimal stationary range, iff $P_\V$ can not be decomposed into a sum of invariant projectors.
\end{itemize}
\end{proposition}

\begin{proof}

$\bullet$ If $P_\V$ is invariant under $\S$, the space $\V$ must be a subspace of $\Rr$, by definition of $\S$.
For any $\rho$ with $\range(\rho)\subset\Rr$ we have then
\beq
\Tr(P_\V\rho)=\Tr(\S^t(P_\V)\rho)=\Tr(P_\V \TT^t(\rho)).    \label{PVinv}
\eeq
This implies, considering $\rho$ with $\range(\rho)\subset\V \neq \{0\}$ that $\V$ is an enclosure.
On the other hand, if $\V$ is an enclosure inside $\Rr$,
preserving the trace implies, again
for any $\rho$ with $\rho=P_\Rr\rho P_\Rr$,
$$
\Tr(P_\V\rho)=\Tr(P_\V \TT^t(\rho))=\Tr(\TT^{t\dag}(P_\V)\rho)=\Tr(\S^t(P_\V)\rho),
$$
implying $\S^t(P_\V)=P_\V$.

$\bullet$ Consider an observable $A=P_\Rr \tdt A\tdt P_\Rr$ and its spectral representation
$$
A=a_1 P_1+a_2P_2+\ldots \quad\quad \textrm{with} \quad \sum_i P_i=P_\Rr,
$$
where $a_1<a_2\ldots$.
If the $P_i$ are invariant, then $A$ is invariant.
On the other hand, $\TT^{t\dag}$ being unital, monotonicity (Theorem \ref{splitting}.\ref{Dspace3})
and the characterization in equation (\ref{HeisenbergR}) imply $\S^t(P_\Rr)=P_\Rr$,
and  $A$ being invariant under $\S$ implies
that also $A-a_1 P_\Rr=(a_2-a_1)P_2+(a_3-a_1)P_3 \ldots$ is invariant.
This operator is zero on the subspace $P_1\Rr$ and strictly positive on $(P_\Rr -P_1)\Rr$.
For $\rho=P_1\rho P_1$ we have
$$
\Tr(\TT^t(\rho)\tdt(A-a_1 P_\Rr))=\Tr(\rho\tdt\S^t(A-a_1 P_\Rr))=0,
$$
implying that $P_1$ projects onto an enclosure and is therefore an invariant operator.
So we know that $A+(a_2-a_1)P_1=a_2(P_1+P_2)+a_3P_3+\ldots$ is invariant, and
iterating in the same way we as above, we infer that
$(P_1+P_2)$, then $P_1+P_2+P_3$ \ldots and therefore all the $P_i$ are invariant.

$\bullet$ The decomposition of a projector as $P_\X=P_\V+P_\W$
is equivalent to the decomposition of the space as $\X=\V\oplus\W$,
and Proposition \ref{equi}.\ref{equi4} gives the proof.
\end{proof}

We remark that this Proposition enables the generalization of a result of Frigerio, \cite{F77},
where a characterization of the set of invariant operators is given,
under the condition that there exists a single stationary state which has the whole space $\Hh$ as its range.
The generalization is here

\begin{cor}
Suppose $\Rr$ can be split in a unique way into  minimal stationary ranges $\U_k$.
Then those operators $A=P_\Rr\tdt A\tdt P_\Rr$ which are invariant under the cut off evolution $\S^t$
form an abelian algebra, namely $\{A=\sum_k a_k P_{\U_k}\}$.
\end{cor}

We have now shown more than we need to give
\begin{proof}\textit{of Theorem \ref{splitting2}}\\
$\Rr$ is an enclosure as stated in Proposition \ref{rspace}.\ref{rspace3}.
So the Theorem is a special case of the last item of Proposition \ref{equi}.\ref{equi4}.
\end{proof}

\subsection{Handling non-orthogonal ranges of stationary states}\label{details3}

Let $\rho$ and $\sigma$ be two different extremal invariant states, with $\V=\rm{range}(\rho)$
not orthogonal to $\range(\sigma)$.
All these subspaces are contained in $\Rr$.
From the mapping of matrix-blocks, Proposition \ref{blockindep}, we infer that $\sigma_\V=P_\V\sigma P_\V$ is stationary,
from uniqueness, Proposition \ref{equi}.\ref{equi2}, we infer $\sigma_\V=\lambda\, \rho$
with $\lambda > 0$.
So the rectangular matrix block $\sigma P_\V$ has as many independent column-vectors as $\rho$,
and $\dim(\rm{range}(\sigma))\geq \dim(\rm{range}(\rho))$.
The opposite inequality follows with the same arguments, and we infer

\begin{lem}  \label{samedim}                                    
If there exist two different extremal invariant states with mutually non-orthogonal
ranges, their ranges have the same dimension.
\end{lem}

The sum of invariant states gives (when normalized) again an invariant state.
The range of this state is an enclosure, and we can infer:
\begin{lem}                                                 
The linear span of two enclosures is an enclosure.
\end{lem}

Now if $\U$ and $\V$ are two different minimal non-orthogonal enclosures,
each with dimension $n$, the dimension of their linear span $\X$
is larger than $n$ but not bigger than $2n$.
Since it is an enclosure it
can be decomposed as $\X=\V\oplus \ldots$ into minimal enclosures -- see Proposition \ref{equi}.\ref{equi4} --
which are all non-orthogonal to $\U$ and therefore (Lemma \ref{samedim}) they have all the same dimension $n$.
So $\dim(\X)$ has to be exactly $2n$.
\begin{lem} \label{RingW}                                 
The linear span $\X$ of two non-orthogonal minimal enclosures $\U$ and $\V$,
each of dimension $n$,
can be decomposed into $\X=\V\oplus\W$, where $\W$ is again a minimal
enclosure of dimension $n$ and is again not orthogonal to $\U$.

\end{lem}
This orthogonal decomposition of $\X$ enables us to find a whole family of
non-orthogonal enclosures:

\begin{proposition}\textbf{Existence of a ring of minimal enclosures}\label{ring}                       

If there exists a pair of non-orthogonal minimal enclosures $\U$ and $\V$,
there exists a whole ring \, $\V(\alpha)\subset \X= lin.span(\U,\V)$,
$\alpha \in (-\pi/2, \,+\pi/2]$, of minimal enclosures,
with $\V=\V(0)$, $\U=\V(\alpha_\U)$ for some $\alpha_\U$, and $\V(\pi/2)=\V^\perp\cap\X$.
Representing $\X$ as $\X= \C^2\otimes\V$ gives
for the set of projectors $P(\alpha)$ onto $\V(\alpha)$ the equivalence relations
\beq\label{ring2}
P(\alpha)= p(\alpha)\otimes\one ,\quad\quad
p(\alpha)=\left(
            \begin{array}{cc}
              \cos^2\alpha & \sin\alpha\cos\alpha \\
              \sin\alpha\cos\alpha & \sin^2\alpha \\
            \end{array}
          \right)\, ,
\eeq
as each $P(\alpha)$ can be decomposed as
\beqa\label{ringP}
P(\alpha)=\cos^2\alpha \tdt P_\V
 + \sin\alpha\cos\alpha \tdt (Q +  Q^\dag) + \sin^2\alpha \tdt P_\W ,\\
{\rm with}\quad Q\tdt Q^\dag=P_\V,\quad Q^\dag\tdt Q=P_\W.\label{ringPP}
\eeqa
\end{proposition}
\begin{proof}
The cut off evolution $\S$ does not mix matrix blocks associated to enclosures (Proposition \ref{blockindep}).
Decomposing the invariant projector $P_\U$ into $P_\U=M_\V+M_C+M_\W$ must therefore
give invariant components.
From  Proposition \ref{sinv} we can infer that $M_\V$ must be some multiple of $P_\V$,
and that $M_\W$ must be some multiple of $P_\W$.
Inserting this decomposition into the equation $P_\U^2=P_\U$ gives a set of equations
for the components, with equations (\ref{ringP}, \ref{ringPP}) as solution, with some $\alpha$.
Also the matrix block $M_C=Q + Q^\dag$, giving the coherence of $\V$ with $\W$, must be invariant.
Therefore, any linear combination of $P_\V$, $P_\W$ and $Q+Q^\dag$ is invariant and
the $P(\alpha)$ are invariant projectors for each $\alpha$.

Choosing a pair of bases, one for $\V$ and one for $\W$, such that the isometry $Q$
which appears in equation (\ref{ringP}) maps one basis onto the other,
representing $\X$ as $\X= \C^2\otimes\V$,
makes $P(\alpha)= p(\alpha)\otimes\one$,
with the $p(\alpha$) given in equation (\ref{ring2}).
\end{proof}

The existence of such a ``ring'' of enclosures is not yet the ultimate wisdom.
Investigating the effects of complete positivity of $\TT$ and $\S$, we will see in Proposition
\ref{sphere} that not only $M_C=Q+Q^\dag$ is invariant, but also $Q$ and $Q^\dag$ separately,
and that there exists a whole sphere of enclosures inside of a $2n$-dimensional subspace $\X$.
But here, in this Section, we study only those effects, which would exist also for
semigroups which are just positive.
And the relation of enclosures ``being in one ring'' suffices to infer a unitary equivalence:

\begin{proposition} \textbf{Unitary relations inside a ring} \label{ringu}                             

For a ring of enclosures with projectors $P(\alpha)$ as in Proposition \ref{ring} there exists a group
of unitary operators $U_\alpha$ and $R_\alpha$ rotating and reflecting the ring,
\beq
R_\alpha:=\one-2P(\alpha),\qquad U_\alpha:=R_{\alpha/2}\tdt R_0
\eeq
\beq
U_\alpha P(\beta)U_\alpha^\dag =P(\alpha+\beta),\qquad R_\alpha\tdt P(\beta)\tdt R_\alpha =P(2\alpha-\beta),
\qquad U_\alpha R_\beta U_\alpha^\dag=R_{\alpha+\beta},
\eeq
giving a symmetry both for $\TT$ restricted to states with range in $\Rr$, and for $\S$:
\beqa
\forall \,\,\rho &&= P_\Rr\rho P_\Rr\quad \nonumber\\
&&\TT^t(U_\alpha \rho U_\alpha^\dag) = U_\alpha \TT^t(\rho) U_\alpha^\dag,
\qquad \TT^t(R_\alpha \rho R_\alpha)=R_\alpha\tdt \TT^t(\rho)\tdt R_\alpha, \\
\forall \,\, A &&= P_\Rr A P_\Rr\quad \nonumber\\
&&\S^t(U_\alpha^\dag A U_\alpha) = U_\alpha^\dag \S^t(A) U_\alpha, \qquad
\S^t(R_\alpha \tdt A \tdt R_\alpha)=R_\alpha\tdt \S^t(A)\tdt R_\alpha.
\eeqa
\end{proposition}
\begin{proof}
Decompose $\Rr=\X\oplus(\X^\perp\cap\Rr)$,
with $\X$ being the $2n$-dimensional subspace supporting the ring.
Moreover,  decompose $\X$ into two mutually orthogonal enclosures, $\X=\V\oplus\W$.
Consider the action of operators only on $\Rr$,
where the  matrix-sections are not being mixed, neither by $\TT$ nor by  $\S$ (Proposition \ref{blockindep}).
So  $R_0=\one-2P(0)$ acting as $\rho\mapsto R_0\tdt\rho\tdt R_0$ is a reflection of the ring,
commuting with the semigroups (restricted to their action on $\Rr$),
since it does nothing but changing the signs of the coherence blocks giving
coherence of $\V$ with the other enclosures orthogonal to $\V$.
The same argument holds for each reflection $R_\alpha=\one-2P(\alpha)$,
just by considering its action in the case of decomposing $\X=\V(\alpha)\oplus\V(\alpha+\pi /2)$.

Representing $\X$ as $\X= \C^2\otimes\V$ and
using $2\times 2$-matrices $^2U_\alpha$ to perform the rotation in $\C^2$,
we see that we may construct the rotations and the reflections of the ring
as $U_\alpha=\, ^2U_\alpha\otimes\one$, and $R_\alpha=U_\alpha R_0 U_\alpha^\dag $.
Since $R_\alpha\tdt R_0=U_{2\alpha}$,
the reflections $R_\alpha$ generate the whole group,
and also the rotations by the $U_\alpha$ commute with the actions of the semigroups
$\S$ and $\TT$ restricted to $\M_\Rr$.
\end{proof}

Transformations by these unitary operators $U_\alpha$ commute with the cut-off semigroups, they represent dynamical symmetries on $\Rr$,
they map stationary states onto stationary states and invariant operators onto invariant operators.
The symmetry-group of one ring of minimal enclosures, formed by these unitary operators,
is equivalent to the group $O(2)$.

There is moreover the possibility that there exist enclosures $\X$ with dimension $m\cdot n$ with $m>2$,
where each pair of minimal enclosures contained in $\X$ is in a ring as we just described.
More details on such structures follow in Corollary \ref{splittingx} and in Subsection \ref{cohblocks}.
But even without such precise knowledge one can now prove the equivalences of decompositions:

\begin{proof} \emph{of Theorem \ref{splitting3}}

Consider two different decompositions into minimal enclosures,
$$\Rr=\bigoplus \V_k \qquad \textrm{and}\qquad \Rr=\bigoplus \W_\ell.$$
If each $\W_\ell$ is identical to some $\V_k$ the decompositions are identical.
Otherwise we have to show the existence of a unitary operator mapping each $\W_\ell$ onto some $\V_k$.
If $\V_k$ and $\W_\ell$ are neither identical nor orthogonal for some pair $(k,\ell)$ -- w.l.o.g. we assume $k=1,\,\ell=1$ --
they have the same dimension and
they span a subspace with  twice that dimension.
Therefore, in each of the two decompositions there must be at least one more enclosure
which is not orthogonal to this space of the other decomposition, w.l.o.g.
these are $\V_2$, not orthogonal to $\W_1$, and $\W_2$, not orthogonal to $\V_1$.
Looking wether there exists one more enclosure $\W_3$ which is not-orthogonal to  $\V_1\oplus\V_2$,
including this and also a $\V_3$ and then iterating, one gets an ordered collection with those properties:

$\V_k$ is not orthogonal to the space spanned by $\W_1,\W_2,\ldots\W_{k-1}$, and

$\W_k$ is not orthogonal to the space spanned by $\V_1,\V_2,\ldots\V_{k-1}$. \\
They all span a subspace
\beq
\X=\bigoplus_1^m\V_k=\bigoplus_1^m\W_k.
\eeq
By Proposition \ref{ringu} there is a unitary operator $U_1$ mapping $\W_1\rightarrow\V_1$,
giving a unitary mapping of states or of observables inside the ring spanned by these two enclosures,
commuting with the cut-off semigroups.
Therefore it maps enclosures onto enclosures, such that
$$
U_1 :\quad \W_1\rightarrow\V_1,\qquad (\X\ominus\W_1)\rightarrow (\X\ominus\V_1).
$$
This unitary operator acts as unity on $\X^\perp$.
If $m=2$, the mapping of $\X$ is already as it shall be, since then
$\W_2=\X\ominus\W_1$ and $\V_2=\X\ominus\V_1$.
Otherwise, if $m>2$, we observe that there exist two different decompositions for $\X\ominus\V_1$:
$$
\X\ominus\V_1=\bigoplus_2^m\V_k=\bigoplus_2^m U_1(\W_k).
$$
So there exists another unitary operator, $U_2$, mapping  $U_1(\W_2)\rightarrow\V_2$ (the indices are ``w.l.o.g.''),
acting as unity on $(\X\ominus\V_1)^\perp$:
$$
U_2 \cdot U_1:\qquad  \W_1\rightarrow\V_1,\qquad \W_2\rightarrow\V_2,\qquad (\X\ominus(\W_1\oplus\W_2))\rightarrow (\X\ominus(\V_1\oplus\V_2)).
$$

This procedure continues, ending with
$$U=U_{m-1}\cdot U_{m-2}\cdots U_1 : \qquad \W_k\rightarrow \V_k , \quad\forall k\leq m.$$

If the decompositions of $\Rr\ominus\X$ into $\V_k$ and $\W_\ell$ are still different,
one continues in the same way, finding another subspace $\X_2\subset\Rr$ where the $\W_\ell$ can be
mapped onto the $\V_k$. One continues, finding $\X_3\subset\Rr$, \ldots , until each $\W_\ell\in\Rr$ is mapped onto some $\V_k$.
\end{proof}

Using the notions used to prove Theorem \ref{splitting3}, we can state an
enhanced version of Theorem \ref{splitting2}.
One just needs only to continue this proof by adding to the two different decompositions
$\{\V_k\}$ and $\{\W_k\}$ all possible decompositions $\{\V_{k,\beta}\}$, which may bring in more
non-orthogonal pairs of minimal enclosures:

\begin{cor} \textbf{A unique splitting}\label{splittingx}         

The space $\Rr$ uniquely splits into orthogonal subspaces $\U_k$ and $\X_\ell$,
where each  $\U_k$ is a minimal enclosure, each $\X_\ell$ allows for different decompositions
into minimal enclosures $\V_{\ell,\alpha}$.
All decompositions of each $\X_\ell=\C^m \otimes\V_{\ell,1}$ are unitarily equivalent,
with unitary transformations commuting with $\S$ and with $\TT$
restricted to $\M_\Rr$:
$$
\Rr=\bigoplus_k \U_k \oplus \bigoplus_\ell\X_\ell,
$$
For $\X_\ell =\bigoplus_{k=1}^m \V_k,\,\, m\geq 2$,
there exist other  decompositions $\X_\ell =\bigoplus_{k=1}^m \W_k$, and there exists
for each such pair of decompositions a unitary $U$,
mapping $\W_k\to\V_k$ such that
\beq
\forall\rho=P_\Rr\rho P_\Rr: \quad \TT^t(U\rho U^\dag)=U\TT^t(\rho)U^\dag .
\eeq
\end{cor}
In proving Theorem \ref{splittingxc} in Subsection \ref{cohblocks},
we  present a complete characterization of the set of stationary states,
accompanied by unitary transformations
commuting with $\TT$ inside each $\X_\ell$.

\section{Coherence and the effects of complete positivity}\label{Sectioncohcp}

\subsection{Coherence-blocks}\label{cohblocks}

There is one fact about the coherence blocks (Definition \ref{matrixblocks}) which
can be inferred directly from the studies of non-orthogonal enclosures:

\begin{lem}\label{nostatcoh}             
There is no stationary coherence between two minimal orthogonal enclosures, unless
they are both in one ring containing non-orthogonal enclosures.
\end{lem}
\begin{proof}
Let $\rho$ and $\sigma$ be the stationary states on the minimal enclosures $\V$ and $\W$,
let $\mu$ be a stationary matrix in the coherence block. $\TT$  commutes with taking the adjoint,
therefore, if $\mu$ is stationary, also $\gamma_1=\mu+\mu^\dag$ and $\gamma_2=i(\mu-\mu^\dag)$ are stationary.
At least one of these self-adjoint matrices is not zero, and
$(\rho+\eps\gamma_j+\sigma)/2$ is, for small $\eps$, a stationary state whose range is not orthogonal to
$\V$ and not orthogonal to $\W$. This can happen only if there exists a common ring for these minimal enclosures.
\end{proof}

For more studies on coherence we need the effects of \textbf{complete positivity}.
The standard example of a map which is positive, but not completely positive,
namely transposition, shows that there are cases of non c.p. evolutions
where a matrix block $\M_C$ can not be split further.
For continuous time, an example is
$\Hh =\C^2$, $\TT^t(\rho)=e^{-t}\rho+(1-e^{-t})\rho^\dag$.
For c.p. evolutions however we start with
enhancing Proposition \ref{blockindep}, showing that each matrix block
$\M_C$ splits into two parts which evolve independently of each other.

\begin{proposition}\textbf{Two independent coherence blocks}\label{cohsplit}                

 In evolutions with completely positive $\TT^t$ the matrix block $\M_C$ between two minimal enclosures $\V$ and $\W$
 splits into two mutually independent parts,
 \beqa
 \M_{C1} = \lspan (\,\{|\psi\rangle\langle\phi|,\, \psi\in\V,\, \phi\in\W\}\,), \\
 \M_{C2} = \lspan (\,\{|\phi\rangle\langle\psi|,\, \psi\in\V,\, \phi\in\W\}\,),
\eeqa
such that  each one is mapped by $\TT$, and also by $\S$, onto itself.
\end{proposition}

\begin{proof}
The action of each completely positive $\TT^t$ can be represented with
Kraus-operators (\cite{K71}):
$$\TT^t(\rho)=\sum_k A_k\rho A_k^\dag\qquad{\rm with}\quad \sum_k A_k^\dag A_k=\one.$$
(We suppress the dependence of $t$ and, in the following, also the dependence of $k$, for simplicity.)
We consider $\rho$ with range in $\V\oplus\W\subset\Rr$ and use a block representation
of the Kraus-operators, restricted to act on such matrices,
$$
A=\left(  \begin{array}{cc} B & C \\ D & F \end{array} \right)\qquad\qquad\rm{with}\quad
B:\,\V \rightarrow\V , \qquad D:\,\V\rightarrow\W\, \ldots\quad .
$$
The action of $A$ onto a state $\sigma$ with range in $\V$ is
$$
A:\quad\left(  \begin{array}{cc} \sigma & 0 \\ 0 & 0 \end{array} \right)\,\mapsto\,
\left(  \begin{array}{cc} . & . \\ . & D\sigma D^\dag \end{array} \right).
$$
For positive $\sigma$ this ``lower right'' part is positive for each $D_k$, mutual annihilation is not possible.
But $\V$ is an enclosure, so that ``nothing gets out'' is possible only if all the $D_k$ are $0$.
In the same way one can see that all the $C_k$ have to be zero, since also $\W$ is an enclosure.

The action of a Kraus-operator on an element of $M_C$  is therefore
$$
A:\quad\left(  \begin{array}{cc} 0 & \gamma \\ \delta & 0 \end{array} \right)\,\mapsto\,
\left(  \begin{array}{cc} 0 & B\,\gamma\,F^\dag \\ F\,\delta\,B^\dag & 0 \end{array} \right).
$$
This shows that the ``upper right'' part of $\M_C$ is not mixed with its ``lower left'' part.
\end{proof}

As a consequence of this further decomposition of each coherence block into two non-mixed parts
we find for each pair of non-orthogonal minimal enclosures a larger family of enclosures
than in Proposition \ref{ring}:

\begin{cor}\textbf{Each ring is in a sphere}\label{sphere}                                  

If there exists a pair of non-orthogonal minimal enclosures $\U$ and $\V$,
there exists a whole sphere \, $\V(\alpha,z)$, $\alpha \in (-\pi/2, \,+\pi/2]$, $z\in\C$, $|z|=1$, of minimal enclosures,
with $\V=\V(0,1)$, $\U=\V(\alpha_\U ,1)$ for some $\alpha_\U$.
The set of projectors $P(\alpha,z)$ onto $\V(\alpha,z)$ is isomorphic to the set of $2\times 2$ matrices
\beq
p(\alpha)=\left(
            \begin{array}{cc}
              \cos^2\alpha & z\tdt\sin\alpha\cos\alpha  \\
              z^\ast\tdt\sin\alpha\cos\alpha & \sin^2\alpha \\
            \end{array}
          \right)
\eeq
and each $P(\alpha, z)$ can be decomposed as
\beqa\label{sphere2}
P(\alpha,z)=\cos^2\alpha \tdt P_\V
 + z\tdt\sin\alpha\cos\alpha \tdt Q +  z^\ast\tdt\sin\alpha\cos\alpha \tdt Q^\dag + \sin^2\alpha \tdt P_\W ,\\
\textrm {with}\qquad \W:=\V(\pi/2 ,1) \perp \V,\quad Q\tdt Q^\dag=P_\V,\quad Q^\dag\tdt Q=P_\W.\label{sphere3}
\eeqa
\end{cor}
\begin{proof}
The non-mixing of matrix blocks by $\TT$ and by its cut off dual $\S$
enables the decomposition of the invariant projector $P_\U$ into $P_\U=M_\V+M_{C1}+M_{C2}+M_\W$
with invariant components.
\end{proof}

Such a sphere is now the maximal set of minimal enclosures in the space spanned by non-orthogonal minimal enclosures $\U$ and $\V$.
To show this we use the equivalence in Proposition \ref{sinv} and state:

\begin{lem}\textbf{Uniqueness of stationary coherence}\label{uniquecoh}

There is, apart from a constant factor, at most one invariant matrix $Q_i$ in each
matrix block $\M_{Ci}$  between two minimal enclosures $\V$ and $\W$.
For an appropriate choice of factors these matrices obey
$Q_2 =Q_1^\dag$, and are isometries between the enclosures.
The mapping $A\mapsto Q_1\tdt A\tdt Q_2$, transforming an operator with range in $\W$
into an operator with range in $\V$,  commutes with the cut off semigroup $\S$.
\end{lem}

\begin{proof}
Such invariant matrices exist only in the coherence block of minimal enclosures $\U$, $\W$, when they are in a ring,
as was stated in Lemma \ref{nostatcoh}.
In proposition \ref{ringu} it was shown that there is a group of unitary operators $U_\alpha$
giving a mapping $Q_i\mapsto U_\alpha\tdt Q_i\tdt U_\alpha^\dag$
which leaves the invariance of observables untouched.

We choose the basis such that these $U_\alpha$ can be written, when
considered as an operator acting only on the subspace spanned by the ring, in block-matrix form as:
\beq\non
U_\alpha=\left(
            \begin{array}{cc}
              \cos\alpha\cdot\one & -\sin\alpha\cdot\one \\
              \sin\alpha\cdot\one & \cos\alpha\cdot\one \\
            \end{array}
          \right).
\eeq
It maps, for $\alpha=\pi/4$
\beq
         \left(
            \begin{array}{cc}
              0 & Q_2 \\
              0 & 0 \\
            \end{array}
          \right),
\eeq\non
to
\beq\non
         \frac12\left(
            \begin{array}{cc}
              Q_2 & Q_2 \\
              Q_2 & Q_2 \\
            \end{array}
          \right).
\eeq
Because of the non-mixing of enclosures the part in the block which maps $\V\rightarrow\V$ 
must be invariant.
But here we know, that there is, up to a constant factor, only one invariant operator,
namely the projector $P_\V=\one$.

So $Q_2={\rm const.}\cdot\one$.
The same argument applies to $Q_1$.
(If one is using another basis, the $\one$ in the off-diagonals
appear as mutually adjoint isometries.)
\end{proof}

\begin{proof}\emph{of Theorem \ref{coherence}: }

(\ref{coherence1}) is stated in Proposition \ref{blockindep}.

(\ref{coherence2}) Stationary coherence between two minimal enclosures $\V$ and $\W$ can exist only
when they are in a ring, as is shown in Lemma \ref{nostatcoh}.
Then, by Proposition \ref{ringu}, there is a unitary dynamical equivalence of all the elements
in the ring, especially between $\V$ and $\W$.

(\ref{coherence3}) The uniqueness of a stationary coherence does not yet imply
the decay of coherence in other cases, when no stationary coherence exists.
There may be an ``eigenmatrix'' $\gamma$ of $\TT$,
with $\TT^t(\gamma)=e^{i\omega t}\gamma$.
Here complete positivity is necessary (Proposition \ref{cohsplit}):
Let $\gamma=P_\V\tdt\gamma\tdt P_\W$, in the coherence block for two minimal enclosures.
Note, that  $\gamma^\dag=P_\W\tdt\gamma^\dag\tdt P_\V$ and $\TT^t(\gamma^\dag)=e^{-i\omega t}\gamma^\dag$.
For the modified evolution
\beq\label{modifyt}
\TT_\omega ^t(\rho):=\TT^t(U_t\rho U^\dag_t),\qquad U_t:=e^{-i\omega t}P_\V +P_{\V^\perp},
\eeq
which is also a completely positive semigroup, since it is composed of c.p. maps, the coherence $\gamma+\gamma^\dag$ is stationary.
The extremal stationary states for $\TT$, $\rho_\V$ and $\rho_\W$, are stationary states for $\TT_\omega$ also.
So, they must be in a ring for $\TT_\omega$ as is stated in Lemma \ref{nostatcoh}.
It follows, using Proposition \ref{ringu}, that there must be a symmetry of the dynamics between the subspaces $\V$ and $\W$:
\beq
\exists\, U:\,\V\rightarrow\W,\quad \forall\sigma=P_\V\tdt\sigma\tdt P_\V:\quad
\TT_\omega^t(\sigma)=\TT_\omega^t(U\tdt\sigma\tdt U^\dag).
\eeq
The evolution $\TT_\omega$ acts on the blocks on the diagonal in the same way as the evolution $\TT$, so
\beq
\TT^t(\sigma)=\TT^t(U\tdt\sigma\tdt U^\dag).
\eeq
\end{proof}
An example for such a situation of dynamical symmetry, allowing for oscillation of coherence, is given in 5.5 of \cite{BN08}.
In the following Subsection \ref{discretesection}
we identify precisely all possibilities of other oscillations in discrete time.
As a consequence we find then that other cases of perpetual oscillation can not appear in continuous time.

Knowing all invariant observables enables us to know precisely the set
of enclosures, the minimal ranges of stationary states.
\begin{proof} \emph{of Theorem \ref{splittingxc}:}

We extend the Corollary \ref{splittingx}:
If $\X_\ell$ is spanned by merely two mutually orthogonal minimal enclosures
there is nothing new to prove. One has just to paste together
Corollary \ref{sphere}, showing the possibility to represent $\X_\ell=\C^2\otimes\V$,
and Lemma \ref{uniquecoh}, showing there are no other stationary states with range in $\X_\ell$.

When there are more than two such orthogonal $\V_\alpha\subset\X_\ell$,
the representation  $\X_\ell=\C^m\otimes\V$, with the representation of
each stationary state in the form of equation (\ref{statrhoform}),
follows the same way.
It remains to show that, in the other direction, \emph{each} density matrix $\sigma$
can appear in (\ref{statrhoform}), that each matrix element can be non-zero.
This means, that there exists a stationary coherence for each pair of minimal enclosures in $\X_\ell$.
Now the relation ``existence of stationary coherence'',
is transitive: If $\V_1$ has a stationary coherence with $\V_2$, then they are together in a ring,
for c.p. maps in a sphere.
The same shall be the case for the pair $\V_2$ and $\V_3$.
Being in a ring allows for permutation of $\V_1$ with $\V_2$, commuting with the cut off dynamical semigroups,
and mapping the coherence block between $\V_2$ and $\V_3$ onto
the coherence block between $\V_1$ and $\V_3$.
\end{proof}

\subsection{Special processes in discrete time}\label{discretesection}

The structuring of Hilbert space is independent of the structure of time, whether time is continuous or discrete,
but some things may happen in discrete time which are otherwise impossible:

\begin{description}
  \item[(i)] Decay, dephasing, dissipation may happen immediately, coming to a limit in finite time.
  \item[(ii)] Perpetual oscillation may occur also inside a minimal enclosure.
\end{description}

Examples are, for $t\in\{1,2,\ldots\}$:

\begin{description}
  \item[(i)] $\TT^1:\quad \rho\mapsto\Tr(\rho)\rho_\infty$, with $\rho_\infty$ positive and $\Tr(\rho_\infty)=1$.
  \item[(ii)] $\TT^1:\quad \rho\mapsto \sum_{k=1}^n |k+1\rangle\langle k|\rho|k\rangle\langle k+1|$,
          on $\Hh$ with basis $\{|k\rangle, k=1 \ldots n\}$, $|n+1\rangle\equiv|1\rangle$.
          (The case with $n=2$ has been presented in \cite{BG07}.)
          Oscillating evolutions appear as
          \beq
          \rho(t)= e^{-2  \pi i \cdot t \cdot q/n}  \left(\frac  {1}{ n} \sum_{k=1}^n|k\rangle e^{2\pi i k\cdot q/n}\langle k|\right).
          \eeq
\end{description}

In the spectrum of $\TT^1$ these possible cases appear in (i) as eigenvalues $\lambda=0$;
in (ii) as eigenvalues $\lambda\neq 1,\,|\lambda|=1$, whose eigenmatrix has a minimal enclosure as range.
In continuous time $\TT^t=e^{-tL}$ ($L$ a ``superoperator'') can not have $0$ as an eigenvalue;
and it is shown in the following, in  Corollary \ref{rotationcor},
that the existence of eigenvalues $\lambda$ with $|\lambda|=1$
is completely covered as described here in Theorems \ref{splitting2}, \ref{splitting3}, \ref{coherence}.

For c.p. semigroups in discrete time there are more oscillations and rotations possible than those described
in the investigations at the end of subsection \ref{cohblocks}:

\begin{proposition}\textbf{Special oscillations in discrete time.}\label{discrete}

Suppose $\TT^t$ are c.p. maps of a Quantum Dynamical System in discrete time $t\in\{1,2,\ldots\}$,
and suppose $\V$ is a minimal enclosure with dimension at least $2$.
$\TT^1$ may have eigenmatrices with $\V$ as their range and corresponding to eigenvalues
$\lambda=\exp(2\pi i\cdot q/m)$, $1\leq q <m\leq \textrm{dim}(\V)$.
    This covers all possible cases of a perpetual oscillation inside minimal enclosures.
\end{proposition}

\begin{proof}
The existence is shown with a modified example on $\Hh$ with basis $\{|k\rangle, k=1 \ldots n\}$:
\beqa
\TT^1:\quad \rho &\mapsto &  \left(\sum_{k=m}^n \langle k|\rho|k\rangle \right) |1\rangle\langle 1| \nonumber\\
&+& \sum_{k=2}^{m-1} |k\rangle\langle k-1|\rho|k-1\rangle\langle k| \nonumber \\
&+& \frac{\langle m-1|\rho|m-1\rangle}{n-m+1}\sum_{k=m}^n |k\rangle\langle k|.
\eeqa
For $m=2$ the second term is absent.
Oscillating eigenmatrices (at $t=0$) are
$$
\rho= \frac  {1}{m} \sum_{k=1}^{m-1}|k\rangle e^{2\pi i k\cdot q/m}\langle k|
  +\frac1{m(n-m+1)} \sum_{k=m}^n |k\rangle\langle k|.
$$

To show that other frequencies are impossible, we analyze $\S$, the cut off evolution of observables,
acting in a minimal enclosure $\V$ with basis $\{|k\rangle\}$.
The eigenvalues of $\TT$ with eigenmatrices supported by $\V$ are the same as the
eigenvalues of $\S$ with eigenmatrices $\gamma$ supported by $\V$.
Suppose $\S^1(\gamma)=z\gamma$, $|z|=1$, $z\neq 1$, $\|\gamma\|=1$.
By complete positivity the Kadison inequality \cite{K52} holds, implying $\S^1(\gamma^\ast \gamma)\geq \S^1(\gamma^\ast)\S^1(\gamma)=\gamma^\ast \gamma$.
Since $\V$ supports an invariant state $\rho$ with $\range(\rho)=\V$,
the equality $\Tr(\rho \S^1(\gamma^\ast \gamma))=\Tr(\TT^1(\rho) \gamma^\ast \gamma)=\Tr(\rho\gamma^\ast\gamma)$
implies $\S^1(\gamma^\ast\gamma)=\gamma^\ast\gamma$. By uniqueness, as stated in Proposition \ref{sinv}, $\gamma^\ast\gamma=P_\V$,
so $\gamma$ acts as a unitary operator on $\V$,
representable as $\gamma =\sum u_k |k\rangle\langle k|$, $|u_k|=1$.
If $\gamma$ is an eigenmatrix, then also $e^{i\alpha}\gamma$
is an eigenmatrix to the same eigenvalue, so we may assume w.l.o.g. $u_1 z^t=1$,
with $t={\rm dim}(\V)$.
So the positive observable $\S^t(A)$, with $A=2P_\V+\gamma+\gamma^\ast$, has norm $\|\S^t(A)\|=4$.
Since  $\S$ preserves positivity and is unital,
it preserves the inequality $\|A\|\cdot\one\geq A$, which implies that the norm
of $A$ cannot increase.
Therefore, for each $\ell\in\{0,1\ldots t\}$ there must be be some $k$ such that $u_k \tdt z^\ell=1$.
These $t+1$ conditions for only $t$ numbers $u_k$ can not all be different,
so there must be some $m\in\{1,\ldots t\}$ such that $z^m=1$.
The case $m=1$ gives stationarity, no oscillation,
and there remain the possibilities of oscillation with $\S^m(\gamma)=\gamma$,
with dual oscillation of some states,  $\TT^m(\sigma)=\sigma$, with  $m\in\{2,\ldots t\}$.
\end{proof}
\begin{cor}\label{rotationcor}
In case of an evolution in continuous time there are no oscillations inside minimal enclosures;
and a coherence between two minimal enclosures, if it does not decay,
can either be stationary or ``rotate'' with one special frequency.
Two different cases for one coherence block can not appear in continuous time.
\end{cor}
\begin{proof}
A semigroup in continuous time contains discrete semigroups, with $t\in \{0,t_0, 2t_0,\ldots\}$.
Oscillations in enclosures with dimension $m$ must therefore have frequencies  $2q\pi/(m\tdt t_0)$, with $q<m$,
for all $t_0$, which is impossible.
Now consider two minimal enclosures $\V$ and $\W$.
The four matrix blocks supported by $\V \oplus\W$
can be represented as $\C^2\otimes \acute{\V}$, with $\acute{\V}\simeq \V$, together with the dynamics.
Suppose there were two different rotating coherences. Then we could change the evolution $\TT$ to $\TT_\omega$
as in (\ref{modifyt}), so that one of the coherences is stationary.
Now, if there were both a stationary and a rotating coherence between $\V$ and $\W$,
there were both a stationary and an oscillating state supported by $\acute{\V}$.
This can not happen for Quantum Dynamical Semigroups in continuous time.
\end{proof}

\section{Conclusion and remarks}

For Markovian time evolutions on Hilbert spaces with finite dimension
we have established a decomposition of the Hilbert space which
concerns the asymptotics of the evolution of states.
Although the evolution does not, in general, preserve orthogonality relations, this decomposition is a splitting
of the Hilbert space into orthogonal subspaces.
One of them is emptied out as time goes to infinity; each one of the other subspaces
is the range of one and onely one stationary state.
To describe this procedure for non-specialists, one may cite the example of Schr\"{o}dinger's cat,
which is in a box, together with an atom with unstable nucleus and a devilish machine \cite{S35}.
At  starting time this system is in a state with range in the subspace of decay.
Now we respect the fact that nothing is perfect
and include in our model a possible failure of the Geiger counter.
One of the possible stationary states appearing in the course of time
is the state where the nucleus had a radioactive decay, but the devilish machine
did not react, the cat stays alive.
The other stationary state is, alas, the decayed atomic nucleus together with the dead cat.
Our investigations give the following facts:
\begin{enumerate}
  \item The decaying state and each one of the stationary states live on mutually orthogonal subspaces.
  \item All coherences are decaying, the cat will definitely be either alive (because of a failure of the machine!) or dead.
\end{enumerate}
Decoherence between different states has to take place, unless the states look alike in all details.

So far  we have identified two important processes of Markovian Quantum Dynamics: \textbf{Decay and Decoherence}.
(This is what we have to emphasize concerning Schr\"{o}dinger's cat paradox: A radioactive decay is the dynamics
of an open system, namely of the atom's nucleus, coupled to the large
external system of Quantum fields, in spite of the box.)
Both of the other standard Markovian Quantum Processes go on inside each single enclosure,
where a stationary state is located: \textbf{Dissipation and Dephasing}.
The cat, whether dead or alive, has positive temperature.
In Statistical Mechanics such a Gibbs state is a mixed state.
(The second law of Thermodynamics might be a consequence of dissipative, mixing processes.
But explanations of these fundamental principles are still under discussion.)
The present investigation gave the following result:
Each state, especially each pure state, with range in one of the enclosures
dissipates into the unique stationary state, if time is continuous.
Only for processes in discrete time a perpetual oscillation may take place.

At this point we remark on naming states ``stationary'':
One speaks of a stationary flow of a fluid or gas, if the overall appearance
does not change in time, in spite of the motion of its parts.
This is the analogy  to a mixed state which gets not changed by the Markovian evolution.
The whole state does not change in time, but, when decomposed into a sum of pure states,
none of these parts is invariant, each pure states dissipates.

For Quantum Dynamics inside an enclosure we can identify Dephasing:
The unique stationary state can be diagonalized.
In this basis all off-diagonal matrix elements of other states have
to disappear in the course of time. This is true also in the case
of perpetual oscillations in discrete time.

Studying the Markovian semigroups does not exclude unitary evolutions.
The decomposition of Hilbert space is a generalization of the spectral
decomposition, and unitary evolution may be part of a Markovian evolution in two ways:
\begin{enumerate}
  \item On the Hilbert space there may be a subspace on which the time evolution is unitary.
In such a subspace the enclosures are one-dimensional, supporting eigenstates of ``energy''.
Decoherence does not take place.
  \item A subspace may be factorizable into two spaces, the dynamics is factorized into
a product of a unitary evolution on one factor-space with a non-invertible Markovian evolution
on the other one.
As an example one can think of a system containing two atoms, where one of the atoms
is decaying, the other one not.
\end{enumerate}

The state space structure investigated here considers the asymptotics of evolution as time goes to infinity.
In \cite{BN08} we have established a structure inside the decaying subspace also,
considering cascades of decay, a causal ordering in time.
This is done for Markovian evolutions in continuous time, described with GKS-Lindblad generators.
For evolutions in discrete time one can think of an analogous structure,
but the methods for analyzing still have to be found.
Such a structure of the decaying space would be essential for a description of
the complete asymptotics of the evolution of observables. In this paper
we analyzed only part of it, the restriction to the subspace supporting non-decaying states.
Completion is still a task to be done.


\end{document}